\documentclass[pdftex,leqno,fleqn,12pts]{llncs}

\usepackage{fullpage}

\usepackage[latin1]{inputenc}   
\usepackage{xspace}

\usepackage{graphicx}           
\usepackage{subfigure}

\usepackage{latexsym}           
\usepackage{amsmath}
\usepackage{amssymb}
\usepackage{amsfonts}

\usepackage{algorithmic}
\usepackage{algorithm}

\usepackage{theorem}            

\newcommand{\old}[1]{{}}

\begin{document}

\title{Geometric Spanners With Small Chromatic 
   Number\thanks{Research partially supported by HPCVL, NSERC, MRI, 
                 CFI, and MITACS.}}
\author{Prosenjit Bose\inst{1} \and Paz Carmi\inst{1} \and Mathieu Couture\inst{1} \and 
Anil Maheshwari\inst{1} \and Michiel Smid\inst{1} \and Norbert Zeh\inst{2}}
\institute{School of Computer Science, Carleton University \and
Faculty of Computer Science, Dalhousie University}

\maketitle

\begin{abstract} 
Given an integer $k \geq 2$, we consider the problem of computing the smallest real number 
$t(k)$ such that for each set $P$ of points in the plane, there exists a $t(k)$-spanner for 
$P$ that has chromatic number at most $k$. 
We prove that $t(2) = 3$, $t(3) = 2$, $t(4) = \sqrt{2}$, and give upper and lower bounds on 
$t(k)$ for $k>4$. 
We also show that for any $\epsilon >0$, there exists a $(1+\epsilon)t(k)$-spanner for $P$ 
that has $O(|P|)$ edges and chromatic number at most $k$.
Finally, we consider an on-line variant of the problem where the points of $P$ are given one after 
another, and the color of a point must be assigned at the moment the point is given.
In this setting, we prove that $t(2) = 3$, $t(3) = 1+ \sqrt{3}$, $t(4) = 1+ \sqrt{2}$, and give 
upper and lower bounds on $t(k)$ for $k>4$. 
\end{abstract}



\section{Introduction} 
Let $P$ be a set of $n$ points in the plane. A \emph{geometric graph} 
with vertex set $P$ is an undirected graph whose edges are line segments 
that are weighted by their Euclidean length. For a real number $t \geq 1$, 
such a graph $G$ is called a \emph{$t$-spanner} if the weight of the 
shortest path in $G$ between any two vertices $p$ and $q$ does not exceed 
$t|pq|$, where $|pq|$ is the Euclidean distance between $p$ and $q$. 
The smallest $t$ having this property is called the \emph{stretch factor} 
of the graph $G$. Thus, a graph with stretch factor $t$ approximates the 
$n \choose 2$ distances between the points in $P$ within a factor of $t$. 
The problem of constructing $t$-spanners with $O(n)$ edges for any given 
point set has been studied intensively; see the book by 
Narasimhan and Smid~\cite{smid07} for an overview. 

In this paper, we consider the problem of computing $t$-spanners whose 
chromatic number is at most $k$, for some given value of $k$. The 
goal is to minimize the value of $t$ over all finite sets $P$ of 
points in the plane. We call a spanner whose chromatic number is at most 
$k$ a $k$-\emph{chromatic spanner}.

\begin{problem}
\label{prob-k-col}
Given an integer $k \geq 2$, let $t(k)$ be the infimum of all real
numbers $t$ with the property that for every finite set $P$ of points in 
the plane, a $k$-chromatic $t$-spanner for $P$ exists. Determine 
the value of $t(k)$.
\end{problem}

Observe that in the definition of $t(k)$, there is no requirement on 
the number of edges of the chromatic spanner. This is not a restriction, 
because, as shown by Gudmundsson \emph{et al.}~\cite{glns-adogg-02}, 
any $t$-spanner for $P$ contains a subgraph with $O(n)$ edges which 
is a $((1+\epsilon)t)$-spanner for $P$.    

We show 
how to obtain a $2$-chromatic $3$-spanner for any point set $P$, 
thus showing that $t(2) \leq 3$. 
We also give an example of a point set $P$ such that any $2$-chromatic 
graph with vertex set $P$ has stretch factor at least three. Thus, 
we have $t(2)=3$. 

Next, we show how to compute a $3$-chromatic $2$-spanner of any point set $P$,
thereby proving that $t(3) \leq 2$. 
We also show, by means of an example, that $t(3) \geq 2$. Thus, 
we obtain that $t(3)=2$.  For $k=4$, we show how to compute a 4-chromatic $\sqrt{2}$-spanner of any point set $P$; 
thus  $t(4) \leq \sqrt{2}$.      
Again by means of an example, we also show that $t(4) \geq \sqrt{2}$. 
Therefore, we have $t(4) = \sqrt{2}$. 

For $k>4$, we are not able to obtain the exact value of $t(k)$. 
Inspired by the \emph{ordered $\Theta$-graph} of 
Bose \emph{et al.}~\cite{bose04a}, we show that 
$t(k) \leq 1 + 2 \sin \frac{\pi}{2(k-1)}$. We also show that the 
vertex set of the regular $(k+1)$-gon gives 
$t(k) \geq 1 / \cos \frac{\pi}{k+1}$. 

In the second part of the paper, we consider an on-line variant of the 
problem where the points of $P$ are given one after another, and the 
color of a point must be assigned at the moment when the point is given; 
thus, later on, the color of a point cannot be changed.
This makes the problem more difficult.
Consequently, the bounds are higher, but still tight for $k=2,3,4$. All our bounds are 
summarized in Table~\ref{tab-summary}.

\begin{problem}
\label{prob-k-colonline}
Given an integer $k \geq 2$, let $t'(k)$ be the infimum of all real
numbers $t$ with the property that for every finite set $P$ of points in 
the plane, which is given on-line, a $k$-chromatic $t$-spanner for 
$P$ exists. Determine the value of $t'(k)$.
\end{problem}

A simple variant of the ordered $\Theta$-graph shows that 
$t'(k) \leq 1 + 2 \sin (\pi/k)$. Thus, we have 
$t'(2) \leq 3$, $t'(3) \leq 1 + \sqrt{3}$ and $t'(4) \leq 1 + \sqrt{2}$. 
Since $t'(2) \geq t(2) = 3$, it follows that $t'(2) = 3$. 
We also give examples showing that $t'(3) \ge 1 + \sqrt{3}$ and 
$t'(4) \ge 1 + \sqrt{2}$. We finally show that, for $k \geq 5$,   
$t'(k) \geq 1 / \cos \frac{\pi}{k}$. 

The rest of this paper is organized as follows: in Section~\ref{section-ellipse}, we define the $t$-ellipse property and show its relationship to our problem. In Section~\ref{section-offline}, we give upper and lower bounds for the off-line problem (Problem~\ref{prob-k-col}). In Section~\ref{section-online}, we give give upper and lower bounds for the on-line problem (Problem~\ref{prob-k-colonline}). In Section~\ref{section-simres-chromatic}, we present simulation results. We conclude in Section~\ref{section-conclusion}. In Table~\ref{tab-summary}, we summarize our results. We now motivate our work.


\begin{table}
\begin{center}
\begin{tabular}{|c|c|c|c|c|}\hline
number of colors & 
         \multicolumn{2}{c|}{$t(k)$ (off-line)} & 
         \multicolumn{2}{c|}{$t'(k)$ (on-line)} \\ 
\hline 
$k$    & lower bound & upper bound & lower bound   & upper bound  \\ 
\hline  
2      & $3$         & $3$         & $3$           & $3$        \\ 
\hline 
3      & $2$         & $2$         & $1+\sqrt{3}$  & $1+\sqrt{3}$\\ 
\hline 
4      & $\sqrt{2}$  & $\sqrt{2}$  & $1+\sqrt{2}$  & $1+\sqrt{2}$ \\ 
\hline 
$k$   & $1/ \cos{\frac{\pi}{k+1}}$ & $1+2\sin{\frac{\pi}{2(k-1)}}$ & $1/ \cos{\frac{\pi}{k}}$ & $1+2\sin{\frac{\pi}{k}}$\\
\hline
\end{tabular}
\end{center}
\caption{Summary of our results.}
\label{tab-summary}
\end{table}

\paragraph{Motivation:} 
In a recent paper, Raman and Chebrolu \cite{raman05} proposed a new
protocol, called 2P, allowing to address rural Internet connectivity
in a low-cost manner using off-the-shelf 802.11 hardware. Since
their infrastructure uses several directional antennae at one node
rather than one single omnidirectional antenna, simultaneous
communications are possible at one node. However, due to
restrictions inherent in the 802.11 standard, backbone nodes have to
communicate with each other using a single channel. While
simultaneous transmissions and simultaneous receptions are 
possible, it is not physically possible for one node to both
transmit and receive at the same time. Therefore, backbone nodes
have to alternate between the send and receive states (see
Figure~\ref{fig-rural}). This forces the backbone to be a bipartite
graph, i.e., to have chromatic number two.

\begin{figure}
 \centering\includegraphics{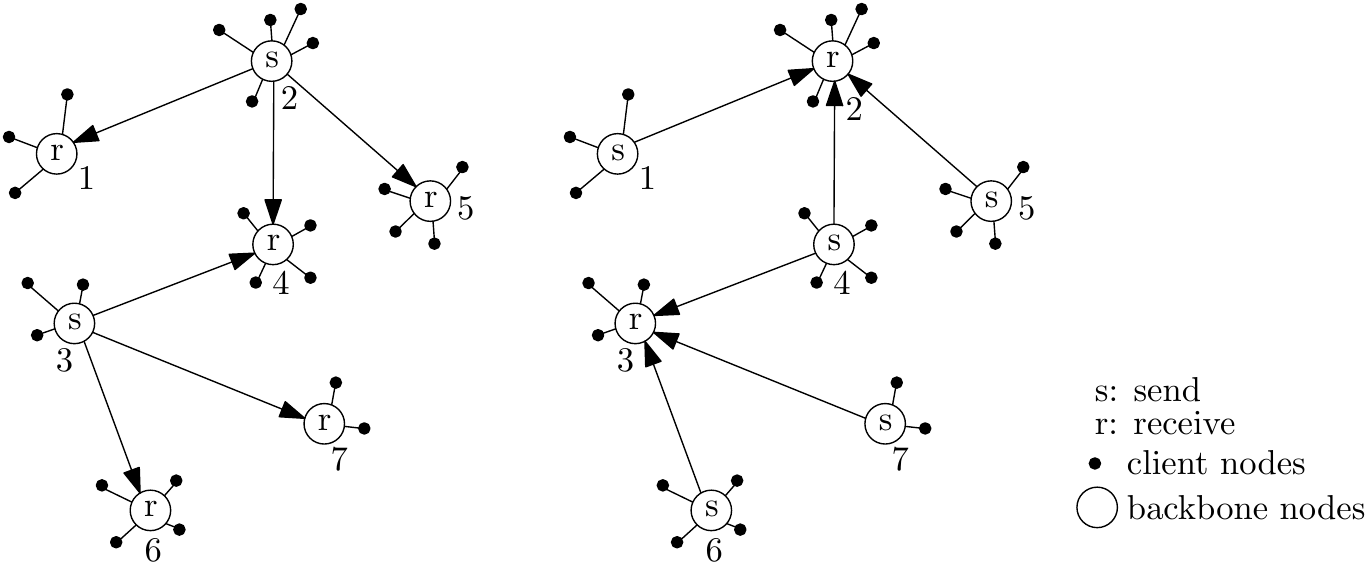} 
 \caption{The two possible states of the backbone nodes.}
 \label{fig-rural}
\end{figure}

The backbone creation algorithm of Raman and Chebrolu \cite{raman05}
outputs a tree, which is obviously bipartite. However, the
tree structure presents the following disadvantage: it is possible
that the path that a message has to follow is much longer than the
distance (either Euclidean or in terms of hops) between the
originating node and its destination. For example, in
Figure~\ref{fig-rural}, a message routed from node~1 to node~3 has
to go through nodes~2 and~4, whereas a direct link between~1 and~3
could be added while still satisfying the bipartition requirement.

Note that the physical constraint preventing nodes to simultaneously
receive and transmit can be met even if the graph is not bipartite.
In fact, any graph with chromatic number $k$ would meet this
requirement: all one has to do is to prevent two nodes that have
different colors to transmit simultaneously. A degenerate case is
when each node has its own color, in which case at most one node can
transmit at any given moment. This case is undesirable, since the
amount of time during which a node can transmit decreases as the
size of the network increases.

For these reasons, it is desirable to have geometric graphs that
have both small chromatic number and small stretch factor.

\section{The $t$-Ellipse Property}\label{section-ellipse}
In this section, we show that Problem~\ref{prob-k-col}, i.e., 
determining the smallest value of $t$ such that a $k$-chromatic 
$t$-spanner exists for any point set $P$, is equivalent to minimizing 
the value of $t$ such that any point set can be colored using $k$ 
colors in a way that satisfies the so-called \emph{$t$-ellipse property}. 

\begin{definition}[$t$-ellipse property] 
Let $k \geq 2$ be an integer, let $P$ be a finite set of points in the 
plane and let $c: P \rightarrow \{1,\ldots, k\}$ be a $k$-coloring of 
$P$. We say that that the coloring $c$ satisfies the 
$t$-\emph{ellipse property} if, for each pair of distinct points
$p$ and $q$ in $P$ with $c(p)=c(q)$, there exists a point $r \in P$
such that $c(r) \neq c(p)$ and $|pr|+|rq|\leq t|pq|$.
\end{definition}

Thus, if $p$ and $q$ have the same color, then the ellipse 
$\{ x \in \mathbb{R}^2 : |px|+|xq|\leq t|pq| \}$ contains a point $r$ of 
$P$ whose color is different from that of $p$ and $q$.   

\begin{proposition}  \label{helloworld}  
Let $k \geq 2$, let $P$ be a set of points in the plane, and let $G$ be 
a $k$-chromatic $t$-spanner of $P$ with $k$-coloring $c$. Then $c$ 
satisfies the $t$-ellipse property.
\end{proposition}
\begin{proof}
Let $p,q\in P$ be two points with $c(p)=c(q)$. Since $G$ is a $t$-spanner, there exists a $t$-spanning
path $\Pi$ in $G$ from $p$ to $q$. Let $r$ be the point on $\Pi$ that is adjacent to $p$. Since the length of $\Pi$ is at most $t|pq|$, we note that $|pr|+|rq|$ is at most $t|pq|$. Since the edge $(p,r)$
is in $G$, it follows that $c(p)\neq c(r)$. Therefore, $c$ satisfies the $t$-ellipse property.\qed
\end{proof}

\begin{proposition}   \label{prop-komplete} 
Let $k \geq 2$, let $P$ be a set of points in the plane, and let 
$c:P\rightarrow\{1,\ldots ,k\}$ be a $k$-coloring of $P$ that
satisfies the $t$-ellipse property. Then, there exists a $k$-chromatic
$t$-spanner of $P$.
\end{proposition}
\begin{proof}
Let $K_c(P)$ be the complete $k$-partite graph with vertex set $P$ 
in which there is an edge between two points $p$ and $q$ if and only if
$c(p)\neq c(q)$. By definition, $K_c(P)$ is $k$-colorable. 
We show that $K_c(P)$ is a $t$-spanner of $P$. Let $p$ and $q$ be two 
distinct points of $P$ such that $(p,q)$ is not an edge in $K_c(P)$. 
This means that $c(p)=c(q)$. Since $c$ has the $t$-ellipse property, 
there exists a point $r$ in $P$ such that $c(r)\neq c(p)$ and 
$|pr|+|rq|\leq t|pq|$. Since $c(r)\neq c(p)$ (and consequently, 
$c(r)\neq c(q)$), the edges $(p,r)$ and $(r,q)$ are both in $K_c(P)$. 
This means that $(p,r,q)$ is a $t$-spanner path in $K_c(P)$ between 
$p$ and $q$. 
\qed
\end{proof}

From this point on, unless specified otherwise, we define the \emph{stretch factor} of a $k$-coloring of 
a point set to be the stretch factor of the complete $k$-partite graph 
induced by this coloring. By Propositions~\ref{helloworld}
and~\ref{prop-komplete}, the problem of determining $t(k)$ is 
equivalent to determining the minimum stretch factor of any 
$k$-coloring of any point set. 
 
We conclude this section by showing why it is sufficient to focus on 
the coloring problem without worrying about the number of edges in 
the spanner. The following theorem is due to 
Gudmundsson \emph{et al.}~\cite{glns-adogg-02}; its proof is based on 
the the well-separated pair decomposition of 
Callahan and Kosaraju~\cite{callahan95}. 

\begin{theorem}  \cite{glns-adogg-02}  \label{thm-wspd} 
Let $\epsilon>0$ and $t \geq 1$ be constants, let $P$ be a set of $n$ 
points in the plane, and let $G$ be a $t$-spanner of $P$. There exists 
a subgraph $G'$ of $G$, such that $G'$ is a $((1+\epsilon)t)$-spanner 
of $P$ and $G'$ has $O(n)$ edges.
\end{theorem}

\begin{proposition} 
\label{prop-gudm}
Let $k \geq 2$, let $P$ be a set of $n$ points in the plane, and let 
$c:P\rightarrow\{1,\ldots ,k\}$ be a $k$-coloring of $P$ that satisfies 
the $t$-ellipse property. Then, for any constant $\epsilon > 0$, there 
exists a $k$-chromatic $((1+\epsilon)t)$-spanner of $P$ that has 
$O(n)$ edges. 
\end{proposition}
\begin{proof} 
By Proposition~\ref{prop-komplete}, there exists a $k$-chromatic 
$t$-spanner $G$ of $P$. By Theorem~\ref{thm-wspd}, $G$ contains a 
subgraph $G'$ with $O(n)$ edges, such that $G'$ is a 
$((1+\epsilon)t)$-spanner of $P$. Since $G$ is $k$-chromatic, 
$G'$ is $k$-chromatic as well. 
\qed  
\end{proof}

\section{Upper and lower bounds on $t(k)$}
\label{section-offline}

The structure of this section is as follows: For $k=2,3,$ and $4$, we 
give coloring algorithms whose outputs have bounded stretch factor. 
Then, we show that these stretch factors are tight by providing point 
sets for which no coloring algorithm can achieve a better stretch 
factor. Then we present our upper and lower bounds for $t(k)$, when $k>4$. 

We now give the coloring algorithm for $k=2$.

\begin{algorithm}
\caption{Offline 2 Colors}\label{alg-off2}
\begin{algorithmic}[1]
\REQUIRE $P$, a set of points in the plane 
\ENSURE $c$, a 2-coloring of $P$ 
\STATE Compute a Euclidean minimum spanning tree $T$ of $P$
\STATE $c \leftarrow$ a 2-coloring of $T$
\end{algorithmic}
\end{algorithm}


\begin{lemma}  \label{prop-ub2off}
For any point set $P$, the 2-coloring computed by 
Algorithm~\ref{alg-off2} has stretch factor at most 3. Thus, we have 
$t(2)\leq 3$.
\end{lemma}
\begin{proof} 
It is sufficient to show that the 2-coloring $c$ computed by
Algorithm~\ref{alg-off2} has the 3-ellipse property. 
Let $p$ and $q$ be two distinct points in $P$ such that $c(p)=c(q)$.  
Observe that $(p,q)$ is not an edge in the minimum spanning tree $T$. 
Let $r$ be the nearest neighbor of $p$. Since the edge $(p,r)$ is in 
$T$, we have $r\neq q$ and $c(r)\neq c(p)$. Since $r$ is closer to $p$ 
than $q$, we have 
\[ |pr|+|rq| \leq |pr|+|rp|+|pq| = 2|pr|+|pq| \leq 2|pq| + |pq|
      =  3|pq| .
\] 
\qed 
\end{proof}

\begin{figure}
\centering\includegraphics{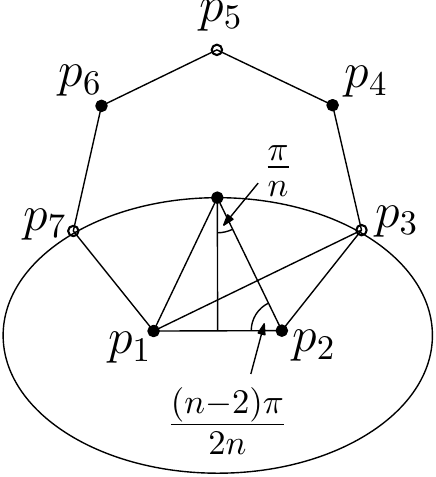}\caption{Lower bound of
$3-\epsilon$ for $k=2$.}\label{fig-lb2off}
\end{figure}

\begin{lemma}\label{prop-lb2off} 
For every $\epsilon>0$, there exists a point set $P$ such that every
$2$-coloring of $P$ has stretch factor at least $3-\epsilon$. 
Thus, we have $t(2)\geq 3$.
\end{lemma}
\begin{proof} 
Let $n$ be an odd integer, and let $P=\{p_1,\ldots,p_n\}$ be the set of 
vertices of a regular $n$-gon given in counter-clockwise order. 
Let $c$ be an arbitrary 2-coloring of $P$. By the pigeonhole principle, 
there are two points in $P$ which are adjacent on the $n$-gon and 
that have the same color. We may assume without loss of generality 
that these two points are $p_1$ and $p_2$. Also, we may assume that 
$|p_1 p_2|=1$ (see Figure~\ref{fig-lb2off}).
Let $t$ be any real number such that $c$ satisfies the $t$-ellipse 
property. Then $|p_1p_3|+1\leq t$. But $|p_1p_3|=2\sin((n-2)\pi/2n)$, 
which tends to 2 as $n$ goes to infinity.
\qed 
\end{proof}

We now consider the case when $k=3$. Our strategy is to construct a 
graph such that any coloring of that graph has the 2-ellipse property. 
We then show that this graph is 3-colorable.

\begin{algorithm}
\caption{Offline 3 Colors}\label{alg-off3}
\begin{algorithmic}[1]
\REQUIRE $P$, a set of $n$ points in the plane
\ENSURE $c$, a 3-coloring of $P$, and $G$, a 3-chromatic graph whose 
vertex set is $P$
\STATE Let $G$ be the graph with vertex set $P$ and whose edge set 
is empty
\STATE Let $e_1,\ldots,e_{\binom{n}{2}}$ be the pairs of points of
$P$ in sorted order of their distances 
\FOR{$i=1$ to $\binom{n}{2}$}
\STATE Let $e_i=(p_i,q_i)$
\IF{$G$ contains no edge $(p,q)$ where 
$|p_ip|+|pq_i|\leq 2|p_iq_i|$ and $|p_iq|+|qq_i|\leq 2|p_iq_i|$}\label{addrule}
\STATE add the edge $e_i$ to $G$
\ENDIF
\ENDFOR
\STATE //assertion: $G$ is 3-colorable (see Lemma~\ref{lemma-3col})
\label{line-assert-3col}
\STATE $c\leftarrow$ a 3-coloring of $G$
\end{algorithmic}
\end{algorithm}



\begin{lemma}   \label{lemma-triangle-free} 
The graph $G$ computed by Algorithm~\ref{alg-off3} is triangle-free.
\end{lemma}
\begin{proof}
Assume that $G$ contains a triangle with vertices $p$, $q$, and $r$. 
We may assume without loss of generality that $(p,r)$ was the last 
edge of this triangle that was considered by the algorithm. This means that $(p,r)$ is the longest edge of the triangle. 
When $e_i=(p,r)=(p_i,q_i)$ in line~4, $G$ already contains the
edge $(p,q)$. Since $|p_ip|+|pq_i|=|pp|+|pr|\leq 2|pr|$ and
$|p_iq|+|qq_i|=|pq|+|qr|\leq 2|pr|$,
the edge $(p,r)$ is not added to $G$. This is a contradiction and, 
therefore, $G$ is triangle-free.
\qed 
\end{proof}

\begin{figure}
  \centering
\begin{tabular}{cc}
\includegraphics{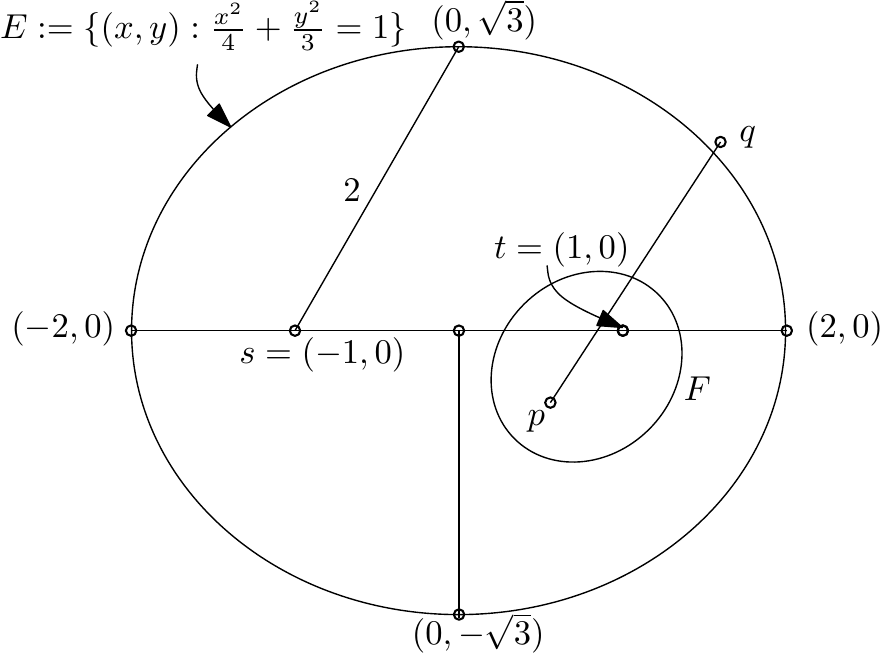} &
\includegraphics{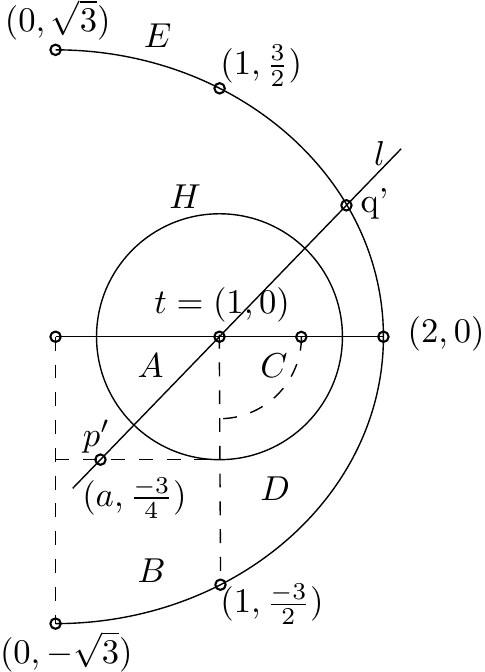}
\end{tabular}
\caption{Proof of Lemma~\ref{lemma-3-planar}.}
\label{fig-3-planar}
\end{figure}

\begin{lemma}  \label{lemma-3-planar} 
The graph $G$ computed by Algorithm~\ref{alg-off3} is plane.
\end{lemma}
\begin{proof}
Assume that $G$ contains two crossing edges $(p,q)$ and $(s,t)$. We may 
assume without loss of generality that $s=(-1,0)$, $t=(1,0)$, and the
pair $(s,t)$ has a larger index than $(p,q)$ after the pairs have been 
sorted in line~2. Thus, we have $|pq| \leq |st|$. When
$e_i=(s,t)=(p_i,q_i)$ in line~4, $G$ already contains the edge $(p,q)$.
Let $E$ be the ellipse whose boundary is determined by the set of points $e$ where $|se|+|et|=2|st|$ (see
Figure~\ref{fig-3-planar}, left). If both $p$ and $q$ are outside $E$,
then $|pq|>|st|$. If both $p$ and $q$ are inside $E$, then it 
follows from step \ref{addrule} of the algorithm that the edge $(s,t)$ is not added to $G$. 
Therefore, exactly one point of $\{p,q\}$ is inside $E$.

Without loss of generality, $p$ is inside $E$, $q$ is outside $E$, 
the pair $(p,t)$ has a smaller index than the pair $(p,s)$ after 
the pairs have been sorted in line~2, and $p$ is below the $x$-axis. 
We will show below that the ellipse $F$ whose boundary is the set of 
points $f$ such that $|pf|+|ft|=2|pt|$ is completely contained inside 
$E$. Thus, since $E$ does not contain any edge, the same is true for 
$F$. It follows that the edge $(p,t)$ is added to $G$, thus preventing 
the insertion of the edge $(s,t)$ because edge $(p,t)$ is easily seen
to have smaller index than $(s,t)$.

It remains to prove that the ellipse $F$ is contained in the ellipse 
$E$. The proof considers four cases, depending on the location of $p$ 
(see Figure~\ref{fig-3-planar}, right).

\vspace{0.5em} 

\noindent 
\textbf{Case A: [$0\leq p_x \leq 1$ and $-3/4 \leq p_y \leq 0$]} We show that for any point $e$ on $E$, we have 
$|pe|+|et| > 2|pt|$. Note that we only need to check the case when 
$e$ is below the $x$-axis and either $p_y=-3/4$ or $p_x=0$. 
We consider these two cases separately.

If $p_y=-3/4$, let $a=p_x$. In this case, we have 
\[ |pe|+|et|-2|pt| = 
    \sqrt{(a-e_x)^2+(3/4+e_y)^2}+\sqrt{(e_x-1)^2+e_y^2} - 
       2\sqrt{(a-1)^2+9/16}.
\] 
Since $e_y=-\sqrt{(12-3e_x^2)}/2$, the above expression is completely 
determined by $a$ and $e_x$. Elementary algebraic transformations (verified with Maple) show that it always
evaluates to a positive value when $e_x$ varies from -2 to 2 and $a$ 
varies from 0 to 1.

If $p_x=0$, let $b=p_y$. We have 
\[ |pe|+|et|-2|pt| = 
    \sqrt{e_x^2+(b-e_y)^2}+\sqrt{(e_x-1)^2+e_y^2}-2\sqrt{1+b^2}.
\] 
As in the previous case, when $e_x$ varies from -2 to 2 and $b$ varies 
from 0 to -3/4, elementary algebraic manipulations (which we verified with Maple) show that the above expression 
evaluates to a positive number.

\vspace{0.5em} 

\noindent 
\textbf{Case B: [$0\leq p_x \leq 1$ and $p_y < -3/4$]} In this case, we show that $|pq|>|st|$. Let $a$ be 
the $x$-coordinate of $p$, let $p'$ be the point $(a,-3/4)$, and let 
$q'$ be the intersection of the line $l$ through $t$ and $p'$ with the 
ellipse $E$. Since $|pq|\geq|p'q'|$, it is sufficient to show that 
$|p'q'| > |st|=2$. The line $l$ is given by the equation 
\[ y = \frac{3(x-1)}{4(1-a)}. 
\] 
Since $E$ is given by the equation $3x^2+4y^2=12$, the intersection 
between $E$ and $l$ is given by:
\[ 4(1-a)^2x^2+3(x-1)^2-16(1-a)^2=0.
\] 
For a fixed value of $a$, let $x(a)$ be the largest root of the above 
polynomial, and let $y(a)$ be the $y$-coordinate of $l$ at $x=x(a)$. 
Then
\[ |p'q'| = \sqrt{(x(a)-a)^2+(y(a)+3/4)^2} . 
\] 
When $a$ varies from 0 to 1, this 
expression always evaluates to strictly more than 2. 

\vspace{0.5em} 

\noindent 
\textbf{Case C: [$p_x>1$ and $|pt| \leq 1/2$]} In this case, the ellipse $F$ is completely contained 
in the circle $H$ centered at $t$ whose radius is $3/2$. Since $H$ is 
contained in $E$, $F$ is also contained in $E$.

\vspace{0.5em} 

\noindent 
\textbf{Case D: [$p_x>1$ and $|pt| > 1/2$]} We separate this case into two subcases, depending on 
whether $q$ has a positive or negative $x$-coordinate. If it is positive, 
then the part of $\overline{pq}$ that is above the $x$-axis has length 
at least $3/2$ and the part of $\overline{pq}$ that is below the $x$-axis 
has length more than $1/2$, which means that $|pq|>2=|st|$. If the 
$x$-coordinate of $q$ is negative but greater than $-1$, then the same 
reasoning applies. If the $x$-coordinate of $q$ is smaller than $-1$, 
then the projection of $\overline{pq}$ on the $x$-axis is larger than 
2, which means that $|pq|>2=|st|$.
\qed 
\end{proof}

\begin{lemma} \label{lemma-3col} 
The graph $G$ computed by Algorithm~\ref{alg-off3} is 3-colorable.
\end{lemma}
\begin{proof}
By Lemmas~\ref{lemma-triangle-free} and Lemma~\ref{lemma-3-planar},
$G$ is plane and triangle-free. It is known that such a graph is
3-colorable; see~\cite{g3color},\cite{thomassen03}.
\qed 
\end{proof}

\begin{lemma}\label{prop-ub3off} 
For any point set $P$, the 3-coloring of $P$ computed by 
Algorithm~\ref{alg-off3} has stretch factor at most 2. Thus, we 
have $t(3)\leq 2$.
\end{lemma}
\begin{proof}
It is sufficient to show that the 3-coloring $c$ produced by
Algorithm~\ref{alg-off3} has the 2-ellipse property. Let $p$ and $q$ be 
points in $P$ such that $c(p)=c(q)$. Let $E$ be the ellipse
whose boundary is the set of points $e$ such that $|pe|+|eq|=2|pq|$.
Since $(p,q)$ is not an edge in $G$, $G$ must contain an edge $(s,t)$
whose two endpoints are inside $E$. Since $c(s)\neq c(t)$, at least
one of $s$ and $t$ has a different color than $p$ and $q$.
Without loss of generality, $s$ is that point. Since $s$ is inside
$E$, we have that $|ps|+|sq|\leq 2|pq|$.
\qed 
\end{proof}

\begin{figure}[h]
\centering\includegraphics{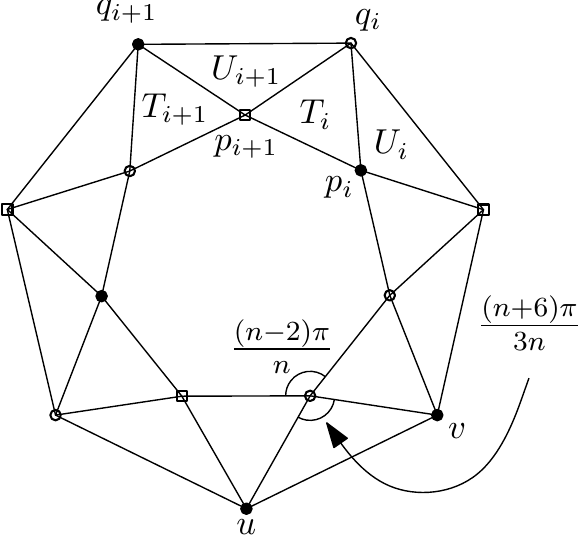}\caption{Lower bound of
$2-\epsilon$ for $k=3$.}\label{fig-lb3off}
\end{figure}

\begin{lemma} 
For every $\epsilon>0$, there exists a point set $P$ such that every
$3$-coloring of $P$ has stretch factor at least $2-\epsilon$. 
Thus, we have $t(3)\geq 2$.
\end{lemma}
\begin{proof} 
Let $n$ be an odd integer, and let 
$P=\{p_1,\ldots,p_n,q_1,\ldots,q_n\}$ where the $p_i$'s are the vertices 
of a regular $n$-gon given in counter-clockwise order, and the $q_i$'s 
are such that the triangles $T_i=(q_i,p_i,p_{i+1})$ are equilateral with interior lying outside the $n$-gon
(indices are taken modulo $n$); see Figure~\ref{fig-lb3off}.
Now consider the set of triangles 
$\mathcal{T}=\{T_1,\ldots,T_n,U_1,\ldots,U_n\}$, where
$U_i=(q_{i-1},p_i,q_i)$. A simple parity argument shows
that, for any 3-coloring of $P$, there is at least one triangle of
$\mathcal{T}$ that has two vertices $u$ and $v$ that are assigned
the same color. If this triangle is a $T_i$, then the stretch factor
between $u$ and $v$ is at least 2. If this triangle is a $U_i$, then
the stretch factor between $u$ and $v$ is at least
$ 1/\sin((n+6)\pi/6n)$, which tends to 2 as $n$ goes to infinity.
\qed 
\end{proof}

%

Next, we consider the case when $k=4$. For this case, we simply use the 
Delaunay triangulation to find a 4-coloring. We then show that this 
coloring satisfies the $\sqrt{2}$-ellipse property.

\begin{algorithm}
\caption{Offline 4 Colors}\label{alg-off4}
\begin{algorithmic}[1]
\REQUIRE $P$, a set of points in the plane 
\ENSURE $c$, a 4-coloring of $P$ 
\STATE Compute the Delaunay triangulation $D$ of $P$ 
\STATE $c\leftarrow$ a 4-coloring of $D$
\end{algorithmic}
\end{algorithm}

\begin{figure}
\centering\includegraphics{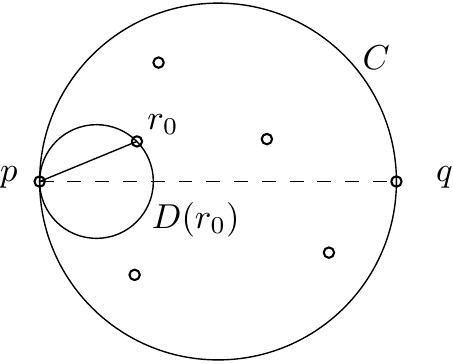}\caption{Upper bound of $\sqrt{2}$ for
$k=4$.}\label{fig-ub4off}
\end{figure}

\begin{lemma}    \label{prop-ub4off} 
For any point set $P$, the 4-coloring of $P$ computed by 
Algorithm~\ref{alg-off4} has stretch factor at most $\sqrt{2}$. 
Thus, we have $t(4)\leq \sqrt{2}$.
\end{lemma}
\begin{proof}
It is sufficient to show that the coloring $c$ computed by
Algorithm~\ref{alg-off4} has the $\sqrt{2}$-ellipse property. 
Let $p$ and $q$ be points of $P$ such that $c(p)=c(q)$. Since $(p,q)$
is not an edge in the Delaunay triangulation, the circle $C$ whose 
diameter is $pq$ contains at least one point of $P$. For a point $r$ 
inside $C$, let $D(r)$ be the circle through $p$ and $r$ whose center 
is on $pq$ (see Figure~\ref{fig-ub4off}). Let $r_0$ be the point inside 
$C$ such that $D(r_0)$ has minimum diameter. Then, $D(r_0)$ is an empty
circle with $p$ and $r_0$ on its boundary, which means that $(p,r_0)$
is a Delaunay edge. Therefore, $c(r_0)\neq c(p)$, and since $r_0$ is
inside $C$, we have $|pr_0|+|r_0q|\leq \sqrt{2}|pq|$.
\qed 
\end{proof}

\begin{lemma} 
For every $\epsilon>0$, there exists a point set $P$ such that every
$4$-coloring of $P$ has stretch factor at least $\sqrt{2}-\epsilon$.
Thus, we have $t(4)\geq \sqrt{2}$.
\end{lemma}
\begin{proof} 
Let $n$ be an odd integer, and let 
$P=\{p_1,\ldots,p_n,q_1,\ldots,q_n\}$, where the $p_i$'s are the 
vertices of a regular $n$-gon, the $q_i$'s are the vertices of a larger 
regular $n$-gon with the same center, and $|q_ip_i|=|p_ip_{i+1}|$ for 
all $i$; refer to Figure~\ref{fig-lb4off}.
Let $Q_i$ be the quadrilateral $(p_i,p_{i+1},q_{i+1},q_i)$. A simple 
parity argument shows that for any 4-coloring of $P$, there is at 
least one $Q_i$ that has two vertices $u$ and $v$ that are assigned the 
same color. The stretch factor between $u$ and $v$ is then at least
$1/ \sin((n+2)\pi/4n)$, which tends to $\sqrt{2}$ when $n$ goes to
infinity.
\qed 
\end{proof}

\begin{figure}
\centering\includegraphics{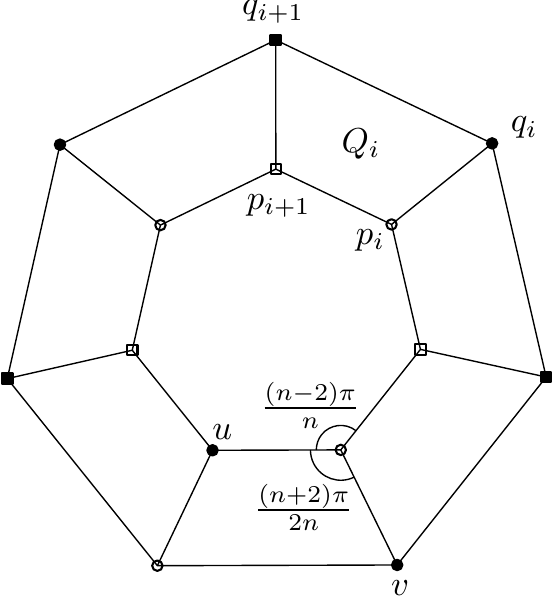}\caption{Lower bound of
$\sqrt{2}-\epsilon$ for $k=4$.}\label{fig-lb4off}
\end{figure}

Our general algorithm for values $k > 4$ uses ideas from the ordered 
$\Theta$-graph of Bose \emph{et al.}~\cite{bose04a}. 
We take advantage of the fact that we are in an off-line context, so that 
we can sort the points according to their $y$-coordinates. We process 
the points one by one from the lowest to the highest, splitting the 
half-plane below the current point $p$ being processed into $k-1$ cones 
of angle $\pi / (k-1)$ and having their apex at $p$. 
For each such cone $c_j$, we take the point $r_j$ in $c_j$ that is 
closest to $p$. Then we assign $p$ a color that has not been assigned to 
any of the $r_j$'s. The fact that this algorithm uses at most $k$ colors 
is straightforward, since there are at most $k-1$ such $r_j$. 

\begin{algorithm}
\caption{Offline $k$ Colors}\label{alg-offk}
\begin{algorithmic}[1]
\REQUIRE $P$, a set of points in the plane 
\ENSURE $c$, a $k$-coloring of $P$ 
\STATE Let $p_1,\ldots,p_n$ be the points of $P$ sorted in non-decreasing 
order of $y$-coordinates 
\FOR{$i=1$ to $n$}
\STATE partition the half-plane below $p_i$ into $k-1$ cones of
angle $\theta=\pi/(k-1)$ and apex $p_i$ 
\STATE for each cone $c_j$, let $r_j$ be the point in $c_j$ that is 
closest to $p_i$ 
\STATE $c(p_i)\leftarrow \min\{l>0:\forall r_j,c(r_j)\neq l\}$ 
\ENDFOR
\end{algorithmic}
\end{algorithm}

\begin{lemma} \label{prop-ubkoff}
For $k>4$, we have $t(k)\leq 1+2\sin(\pi/(2k-2))$.
\end{lemma}
\begin{proof}
Let $p$ and $q$ be points of $P$ such that $c(p)=c(q)$. We may assume 
without loss of generality that $q_y\leq p_y$. Let $c$ be 
the cone with apex at $p$ that contains $q$ in line~4 of 
Algorithm~\ref{alg-offk}, let $r$ be the nearest neighbor of $p$ in $c$, 
let $r'$ be the intersection between the ray emanating from $p$ through $r$ 
and the circle centered at $p$ with radius $|pq|$,
and let $\alpha=\angle rpq$ (see Figure~\ref{fig-ubkoff}). Then,
\[ |pr| + |rq| \leq |pr| + |rr'| + |r'q| = |pq| + |r'q| =  
         |pq| + 2\sin\frac{\alpha}{2}|pq|
     \leq (1+2\sin\frac{\pi}{2(k-1)})|pq| . 
\] 
It follows that the coloring computed by Algorithm~\ref{alg-offk} 
has the $(1+2\sin(\pi/(2k-2))$-ellipse property. The result follows 
from the fact that $c(r)\neq c(p)$ and that Algorithm~\ref{alg-offk} 
uses at most $k$ colors.
\qed 
\end{proof}

\begin{figure}
\centering\includegraphics{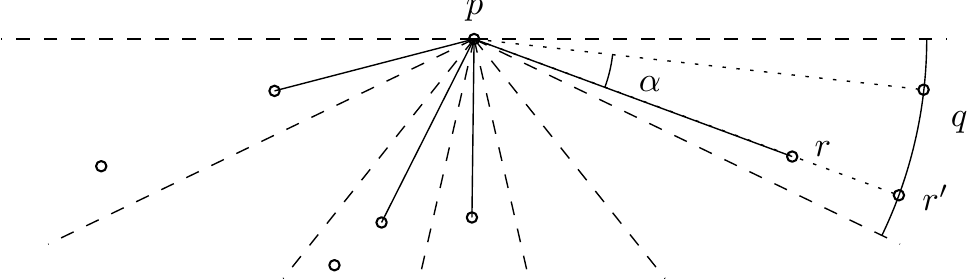}
\caption{Upper bound of $1+2\sin(\pi/(2k-2))$ for $k>4$.}
\label{fig-ubkoff}
\end{figure}

\begin{figure}
\centering\includegraphics{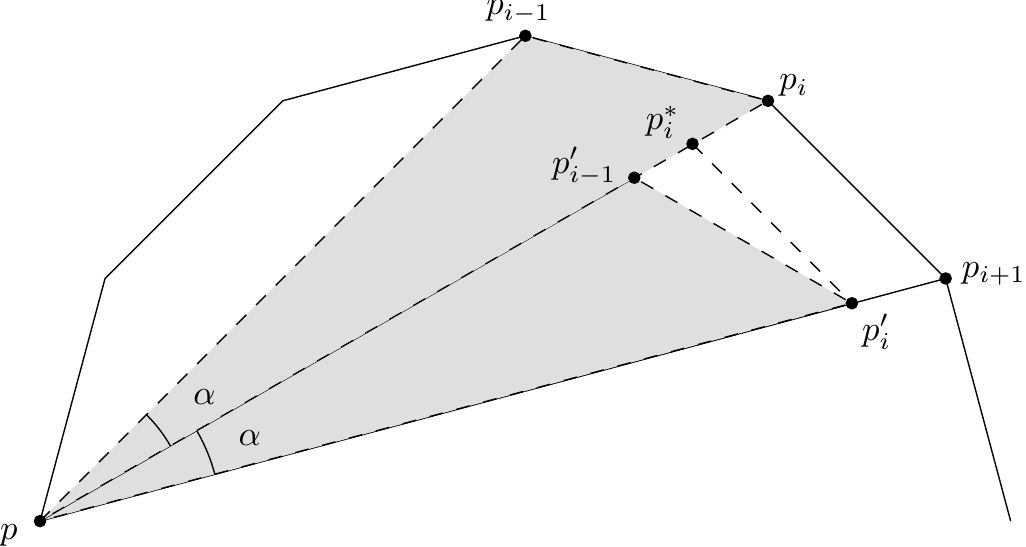}
\caption{Illustration of the proof of Lemma~\ref{lemma-lowerk}.}
\label{fig-lbkoff}
\end{figure}

\begin{lemma}\label{lemma-lowerk} Let $p,q,r$ be three distinct vertices of a regular $(k+1)$-gon. 
Then the ratio $(|pr|+|rq|)/|pq|$ is at least $1/ \cos(\frac{\pi}{k+1})$ and this value is achieved  
when $p$, $r$, and $q$ are consecutive vertices. 
\end{lemma}
\begin{proof} 
For fixed $p$ and $q$, the ratio $(|pr|+|rq|)/|pq|$ is minimized
when $r$ is adjacent to either $p$ or $q$. In that case, the angle $\alpha=\angle qpr=\pi/(k+1)$.
We show that for a fixed point $p$ and three consecutive vertices $p_{i-1},p_i$ and $p_{i+1}$
of the regular $(k+1)$-gon such that $|pp_{i-1}|<|pp_i|<|pp_{i+1}|$ (see Figure~\ref{fig-lbkoff})
the ratio $(|pp_{i-1}|+|p_{i-1}p_i|)/|pp_i|$ is smaller than $(|pp_i|+|p_{i}p_{i+1}|)/|pp_{i+1}|$ and
the result follows.

Without loss of generality, $p_{i-1},p_i$ and $p_{i+1}$ are in clockwise order. Let $p_{i-1}'$ and $p_i'$
be the rotation of $p_{i-1}$ and $p_i$ around $p$ by a clockwise angle of $\alpha$. Also, let $p_i^*$
be the intersection of $\overline{pp_i}$ with the parallel line to $\overline{p_ip_{i+1}}$ 
through $p_i'$. Triangle $pp_i^*p_i'$ is similar to triangle $pp_ip_{i+1}$.
Therefore, 
\begin{eqnarray*}
 (|pp_i|+|p_{i}p_{i+1}|)/|pp_{i+1}| & = & (|pp_i^*|+|p_i^*p_{i}'|)/|pp_{i}'| \\
                                    & = & (|pp_{i-1}'|+|p_{i-1}'p_i^*| +|p_i^*p_{i}'|)/|pp_{i}'|  \\
                                    & > & (|pp_{i-1}'|+|p_{i-1}'p_i'|)/|pp_i'|\\
                                    & = & (|pp_{i-1}|+|p_{i-1}p_i|)/|pp_i|.
\end{eqnarray*}
Therefore, the ratio $(|pp_{i-1}|+|p_{i-1}p_i|)/|pp_i|$ is minimized when $p_{i-1}$ is adjacent to $p$. 
\qed 
\end{proof}

\begin{lemma}  \label{prop-lbkoff} 
For $k>4$, we have $t(k)\geq 1/ \cos(\frac{\pi}{k+1})$.
\end{lemma}
\begin{proof} 
Let $P=\{p_1,\ldots,p_{k+1}\}$ be the vertex set of a regular 
$(k+1)$-gon. 
By Lemma~\ref{lemma-lowerk}, for any three distinct 
points $p$, $q$, and $r$ in $P$, 
the ratio $(|pr|+|rq|)/|pq|$ is at least $1/ \cos(\frac{\pi}{k+1})$ and 
this value is achieved when $p$, $r$, and $q$ are 
consecutive vertices. 

By the pigeonhole principle, any $k$-coloring of $P$ has to assign the 
same color to at least two points, say $p$ and $q$. By the argument 
above, the stretch factor between $p$ and $q$ is at least 
$1/ \cos(\frac{\pi}{k+1})$.
\qed 
\end{proof}

%
The constructions we have shown in this section use a quadratic number of edges since we 
consider the complete $k$-partite graph induced by the coloring of the points. 
To reduce this to a linear number of edges we apply Proposition~\ref{prop-gudm}, 
which slightly increases the stretch factor, giving us the following:

\begin{theorem}
The following are true:
\begin{enumerate}
\item 
For any point set $P$ in the plane, the complete $k$-partite graph induced by the $k$-coloring 
of $P$ computed by the above algorithms has a stretch factor at most
$3$, $2$, $\sqrt{2}$, and $1+ 2 \sin \frac{\pi}{2(k+1)}$ for $k=2$, $k=3$,
$k=4$, $k>4$, respectively. 

\item
For any $\epsilon>0$, there exist point sets such that no coloring algorithm can 
compute a $k$-coloring that has the $t$-ellipse property for $t$ smaller than
$3-\epsilon$, $2-\epsilon$, $\sqrt{2}-\epsilon$, and $1/ \cos \frac{\pi}{k+1}$ 
for $k=2$, $k=3$, $k=4$, $k>4$, respectively.

\item
Thus, we have
$t(2)= 3$, $t(3)= 2$, $t(4) = \sqrt{2}$, and 
$ 1+ 2 \sin \frac{\pi}{2(k+1)} \geq  t(k) \geq 1/ \cos \frac{\pi}{k+1}$ 
for $k > 4$.

\item
It is possible to obtain a $((1+\epsilon)t(k))$-spanner that has $O(|P|)$ edges, 
from 
the coloring computed by the above algorithms.

\end{enumerate}
\end{theorem}


%

\section{Upper and lower bounds on $t'(k)$}
\label{section-online}
 
Recall that in the on-line setting, the algorithm 
receives the points of $P$ one at a time and assigns a color to a point as soon as it receives it.
It cannot change the color of a point after this assignment. Naturally, this setting is more difficult which is
reflected by higher bounds for $t'(3)$ and $t'(4)$. However, we are still able to give the exact value of $t'(k)$ 
for $k=2,3,4$ and provide upper and lower bounds when $k>4$.
In the online setting, we actually provide a general algorithm that is the same 
for all values of $k\geq 2$. Although it is similar to 
Algorithm~\ref{alg-offk}, there are at least two important differences. 
First, since we are in an on-line setting, we cannot process the points 
in the order of their $y$-coordinates. Therefore, we have to use cones 
with an angle greater than $\pi/(k-1)$. If we choose the cones a priori as we do in 
Algorithm~\ref{alg-offk}, we obtain cones whose angle is $2\pi/(k-1)$. 
However, by aligning the cone's bisectors on the points that are chosen 
to be neighbors, we can get a slightly better stretch factor, since in this case, the angle
is reduced to $2\pi/k$. 

\begin{algorithm}
\caption{Online $k$ Colors}\label{alg-onk}
\begin{algorithmic}[1]
\REQUIRE $P$, an arbitrarily ordered list of points in the plane
\ENSURE $c$, a $k$-coloring of $P$
\STATE Let $p_1,\ldots,p_n$ be the points of $P$ in the given order
\FOR{$i=1$ to $n$}
\STATE $P_i\leftarrow \{p_1,\ldots,p_{i-1}\}$
\STATE $j\leftarrow 0$
\WHILE{$P_i\neq\emptyset$}

\STATE $j\leftarrow j+1$

\STATE $r_j\leftarrow$ a nearest neighbor of $p_i$ in $P_i$

\STATE $P_i\leftarrow P_i\setminus \{r_j\}$

\FOR{each $q\in P_i$}

\IF{$\angle qp_ir_j\leq 2\pi/k$}

\STATE $P_i\leftarrow P_i\setminus \{q\}$

\ENDIF

\ENDFOR

\ENDWHILE

\STATE $c(p_i)\leftarrow \min\{l>0 \mid \forall r_j,c(r_j)\neq l\}$

\ENDFOR
\end{algorithmic}
\end{algorithm}

\begin{lemma} \label{prop-ubkon}
For $k\geq 2$, Algorithm~\ref{alg-onk} computes a $k$-coloring that has 
the $t$-ellipse property for $t=1+2\sin(\pi/k)$.
Thus, we have $t'(k) \leq 1+2\sin(\pi/k)$.
\end{lemma}
\begin{proof}
Algorithm~\ref{alg-onk} produces a $k$-coloring, because each $p_i$
selects at most $k-1$ points $r_j$. If there were more than $k-1$
such points, then two of them would form an angle of $2\pi/k$ or
less around $p_i$. However, this situation cannot occur because of
lines~10 and~11. The proof on the stretch factor is 
identical to the one given in
Lemma~\ref{prop-ubkoff}.
\qed 
\end{proof}

\begin{corollary} 
We have $t'(2) \leq 3$, $t'(3) \leq 1+\sqrt{3}$ and 
$t'(4) \leq 1+\sqrt{2}$.  
\end{corollary}

Since an off-line lower bound also provides an on-line lower bound,
we have $t'(2) \geq t(2) = 3$. It follows that $t'(2)=3$. 
We now prove that Algorithm~\ref{alg-onk} is also optimal for $k=3$ and
$4$.

\begin{figure}
\centering\includegraphics{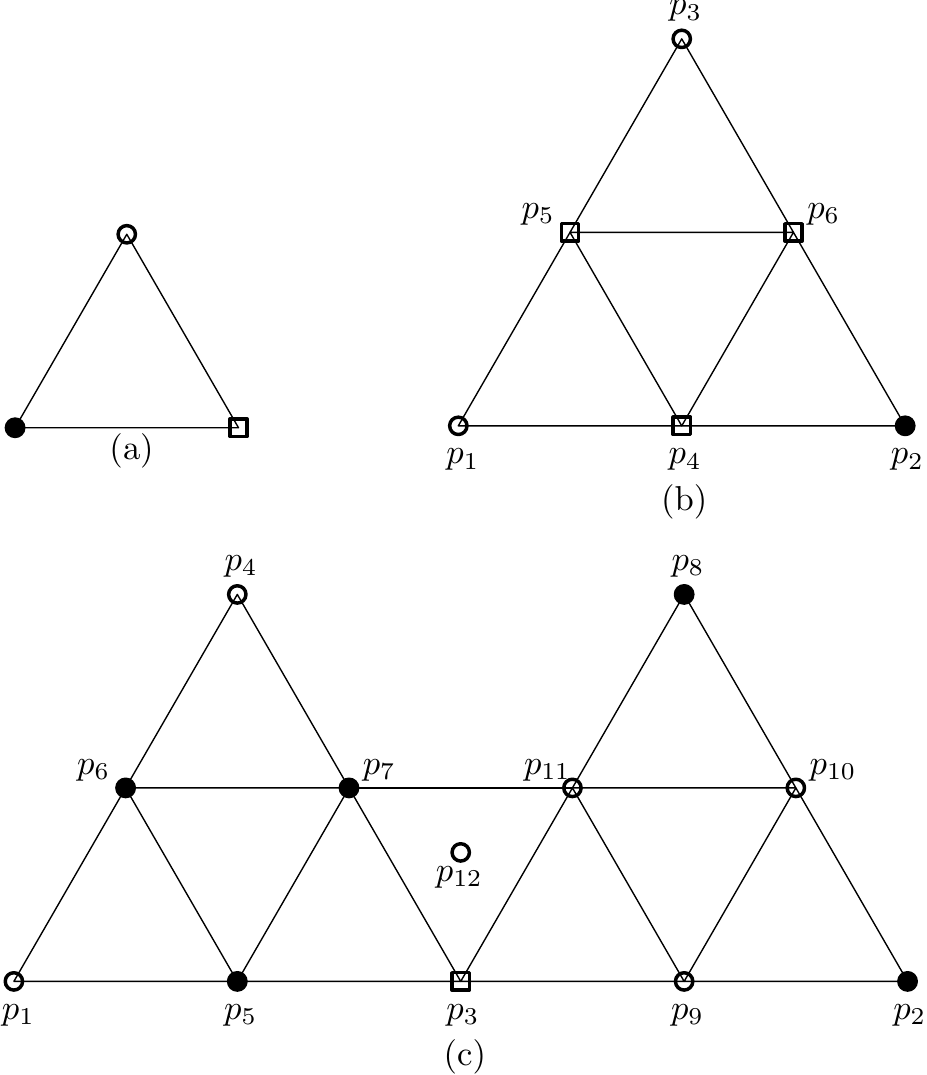}\caption{Online lower bound of
$1+\sqrt{3}$ for $k=3$.}\label{fig-lb3on}
\end{figure}

\begin{lemma} 
Let $\mathcal{A}$ be an arbitrary on-line coloring algorithm
that guarantees a 3-coloring that has the $t$-ellipse property. Then
its stretch factor, $t$, is at least $1+\sqrt{3}$.
\end{lemma}
\begin{proof}
The proof is by an adversarial argument, where the adversary forces 
a stretch factor of at least $1+\sqrt{3}$. The main objective of the 
adversary is to force $\mathcal{A}$ to assign different colors to 
the vertices of an equilateral triangle. Then, the next point is 
placed in the center of this triangle (see Figure~\ref{fig-lb3on}(a)). 
This results in a stretch factor of $1+\sqrt{3}$.

Consider Figure~\ref{fig-lb3on}(b), where the points are numbered by 
the order of insertion. Up to symmetry, there is only one way to 
assign colors to points $p_1$ to $p_6$ such that $t<1+\sqrt{3}$ 
(e.g., $c(p_1)=red, c(p_2)=blue, c(p_3)=red, c(p_4)=green, c(p_5)=green, c(p_6)=green$).
The key property is that the points $p_3, p_4$ and $p_5$ must be assigned the same color that is different from the colors assigned 
to the first three points. If any of these conditions is violated, then the spanning ratio is at least $1 + \sqrt{3}$.

Next, consider Figure~\ref{fig-lb3on}(c), where the point set of
Figure~\ref{fig-lb3on}(b) is reproduced twice. 
Consider triangles $\triangle(p_3,p_5,p_7)$, $\triangle(p_3,p9,p_{11})$ and 
$\triangle(p_3,p_7,p_{11})$ after the insertion of 
$p_{11}$.
At least one of these triangles has to be assigned three different 
colors, otherwise, the stretch factor would already be $1+\sqrt{3}$.
Assume w.l.o.g. that triangle $\triangle(p_3,p_7,p_{11})$ is assigned three
different colors then by the insertion of point $p_{12}$ in the center of the
triangle, we force a spanning ratio of $1 + \sqrt{3}$, as required. 

\qed 
\end{proof}

\begin{figure}
\centering\includegraphics{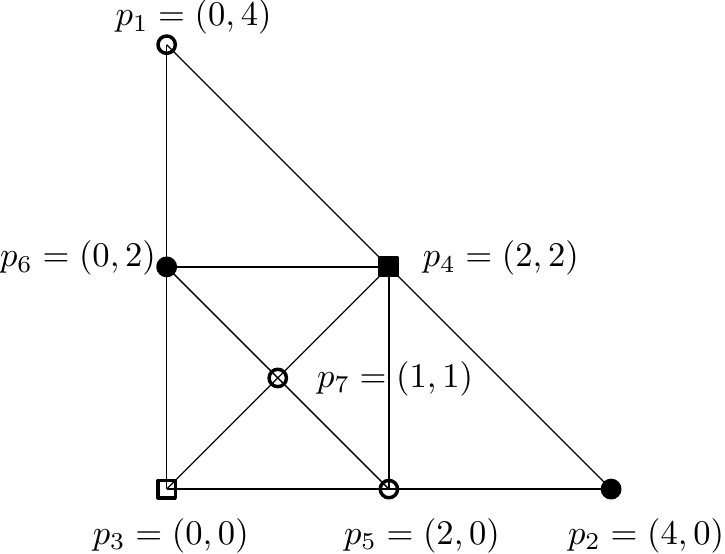}\caption{Online lower bound of
$1+\sqrt{2}$ for $k=4$.}\label{fig-lb4on}
\end{figure}

\begin{lemma} 
Let $\mathcal{A}$ be an arbitrary on-line coloring algorithm
that guarantees a 4-coloring that has the $t$-ellipse property. Then
the stretch factor, $t$ is at least $1+\sqrt{2}$.
\end{lemma}
\begin{proof} 
Consider the point set depicted in Figure~\ref{fig-lb4on}.
$\mathcal{A}$ must assign different colors to $p_3,p_4,p_5$ and $p_6$, otherwise the stretch factor will already be greater than $1+\sqrt{2}$. 
Upon introduction 
of $p_7$, $\mathcal{A}$ must assign it the same color as one of 
$p_3,p_4,p_5$ or $p_6$. The stretch factor between $p_7$ and that 
point is $1+\sqrt{2}$. 
\qed 
\end{proof}

\begin{lemma} 
Let $\mathcal{A}$ be an arbitrary on-line coloring algorithm that 
guarantees a $k$-coloring that has the $t$-ellipse property. Then the stretch factor, $t$, is at 
least $1/\cos(\frac{\pi}{ k})$.
\end{lemma}
\begin{proof} 
Let $P=\{p_1,\ldots,p_{k},q\}$, where the $p_i$' are the vertices of
a regular $k$-gon $K$ and $q$ is the center of the circumcircle of
$K$. If, after processing $p_1$ to $p_k$, $\mathcal{A}$ assigned the
same color to two points, then as in Lemma~\ref{prop-lbkoff},
the stretch factor is $1/ \cos(\frac{\pi}{ k})$. Otherwise, all $p_i$ are
assigned different colors. When $q$ is introduced, the color
$\mathcal{A}$ assigns to it has already been assigned to some other
point $p$. In that case, the stretch factor for the edge $(q,p)$ is
$1+4\sin(\pi/ 2k)>1/ \cos(\frac{\pi}{k})$.
\qed 
\end{proof}


The constructions we have shown in this section use a quadratic number of edges since we 
consider the complete $k$-partite graph induced by the coloring of the points. 
To reduce this to a linear number of edges we apply Proposition~\ref{prop-gudm}, 
which slightly increases the stretch factor, giving us the following:

\begin{theorem}
The following are true:
\begin{enumerate}
\item
For any sequence $P$ of points in the plane, the complete $k$-partite graph 
induced by the on-line $k$-coloring of $P$ computed by the above algorithms 
has a stretch factor at most
$3$, $1+ \sqrt{3}$, $1+ \sqrt{2}$, and 
$1+ 2 \sin\frac{\pi}{k}$ for $k=2$, $k=3$, $k=4$, $k>4$, respectively. 

\item
For any $\epsilon>0$, there exist point sets such that no on-line coloring algorithm can 
compute an on-line $k$-coloring that has the $t$-ellipse property for $t$ smaller than
$3-\epsilon$, $1+ \sqrt{3}-\epsilon$, $1+ \sqrt{2}-\epsilon$, and 
$1/ \cos \frac{\pi}{k} $ 
for $k=2$, $k=3$, $k=4$, $k>4$, respectively.

\item
Thus, we have
$t'(2)= 3$, $t'(3)= 1+ \sqrt{3}$, $t'(4) = 1+ \sqrt{2}$, and 
$ 1+ 2 \sin\frac{\pi}{k}  \geq t'(k) \geq 1/ \cos \frac{\pi}{k}$ for $k > 4$.

\item
It is possible to obtain a $((1+\epsilon)t'(k))$-spanner that has $O(|P|)$ edges, from 
the coloring computed by the above algorithms.

\end{enumerate}
\end{theorem}

\old{
\begin{theorem}
For any sequence $P$ of points in the plane, the $k$-coloring $c$ computed by
Algorithm~\ref{alg-onk} has the $t'(k)$-ellipse property, with 
$t'(2) \leq 3$, $t'(3) \leq 1+ \sqrt{3}$, $t'(4)  \leq 1+ \sqrt{2}$, and
$t'(k) \leq 1+ 2 \sin\frac{\pi}{k}$ for 
$k \geq 5$.
Moreover, there exist point sets such that no on-line coloring
algorithm can compute a $k$-coloring that has the $t$-ellipse property
for $t$ smaller than 
$t'(2) \geq 3$, $t'(3) \geq 1+ \sqrt{3}$, $t'(4)  \geq 1+ \sqrt{2}$, and
$t'(k) \geq 1/ \cos \frac{\pi}{k}$
 for $k \geq 5$. By Theorem~\ref{thm-wspd}, it is possible to
post-process the complete $k$-partite graph induced by the
coloring computed by Algorithm 5 in order to obtain a
$((1+\epsilon)t'(k))$-spanner that has $O(|P|)$ edges.
\end{theorem}
}

\section{Simulation Results}\label{section-simres-chromatic}

Using simulation, we now provide estimates of the average stretch factor of the colorings produced by Algorithm~\ref{alg-offk} and Algorithm~\ref{alg-onk}. Using a uniform distribution, we generated 200 sets of 50 points and 200 sets of 200 points. For each point set, we computed the stretch factor for $k$ ranging from 2 to 10. Figure~\ref{fig-simres-chromatic-50} and Figure~\ref{fig-simres-chromatic-200} show the results we obtained for the stretch factor. The 95\% confidence interval for these values is $\pm 0.0365$.

The general behavior of the average case performance ratio of these algorithms is not much different than what can be expected from the worst case analysis. In particular, the off-line algorithm performs significantly better than the on-line algorithm. Also, in both cases, as $k$ increases, the stretch factor reduction becomes less and less important. Another interesting observation is that for $k$ large enough ($k>6$ for 50 points and $k>3$ for 200 points), the average case stretch factor of the on-line algorithm is worse than the worst case stretch factor of the off-line algorithm.

For $k=2,3$ and $4$, in the off-line case, we used the algorithm for general values of $k$. It is interesting to notice that for $k=4$, the average stretch factor obtained using Algorithm~\ref{alg-offk} is greater than the worst case stretch factor obtained using Algorithm~\ref{alg-off4}. However, Algorithm~\ref{alg-off4} is less practical, since we have to compute a 4-coloring of a planar graph.

\begin{figure}
\centering
\includegraphics[scale=0.67]{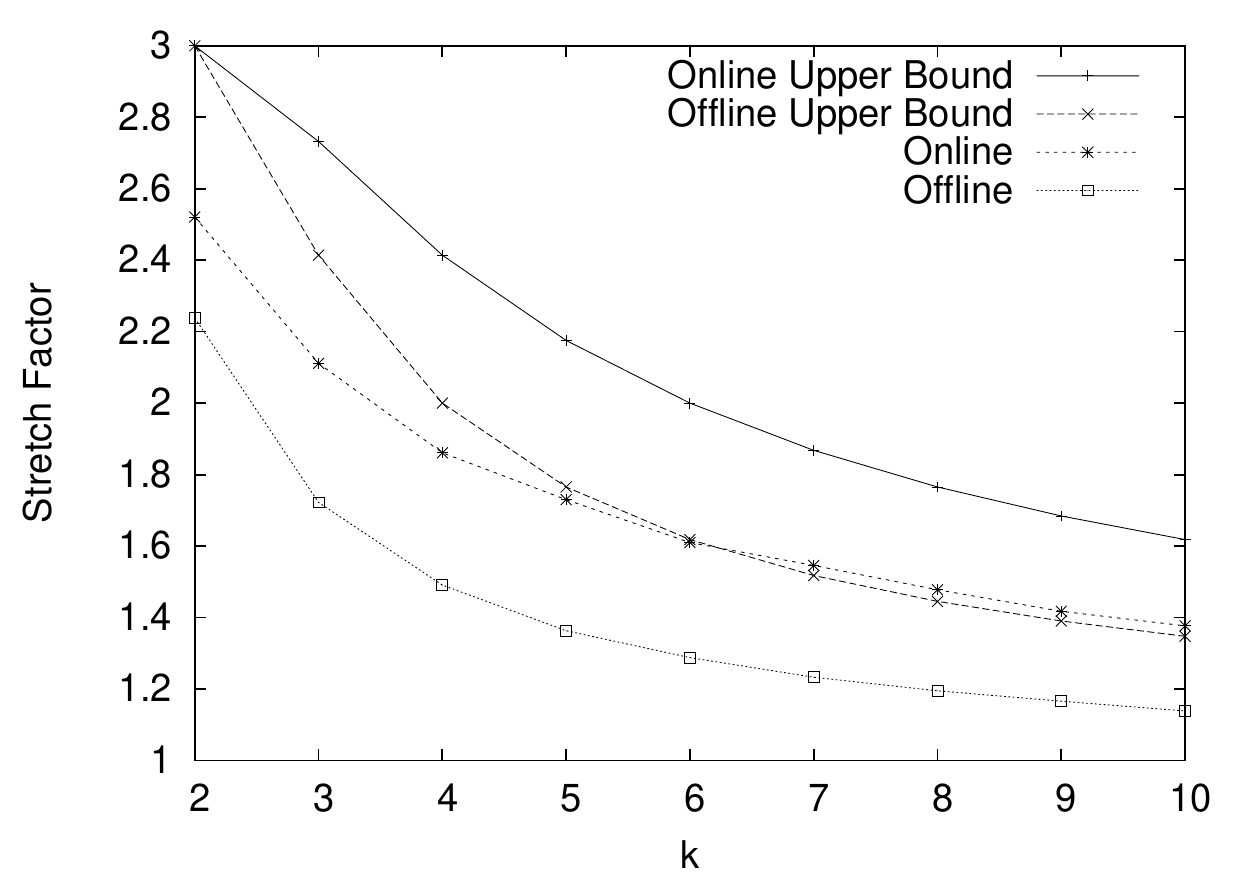}  
\begin{center}
\begin{tabular}{|l|ll|l|ll|}
\hline
$k$ & offline & online & $k$ & offline & online \\
\hline
2	 & 2.2383	 & 2.5208 & 7	 & 1.2329	 & 1.5456\\
3	 & 1.7219	 & 2.1111 & 8	 & 1.1947	 & 1.4778\\
4	 & 1.4907	 & 1.8608 & 9	 & 1.1658	 & 1.4175\\
5	 & 1.3631	 & 1.7300 & 10	 & 1.1384	 & 1.3765\\
6	 & 1.2877	 & 1.6098 & & & \\
\hline
\end{tabular}
\end{center}
%
\caption{Simulation results for 50 nodes using Algorithm~\ref{alg-offk} and Algorithm~\ref{alg-onk}.}\label{fig-simres-chromatic-50}
\end{figure}
\begin{figure}
\centering
\includegraphics[scale=0.67]{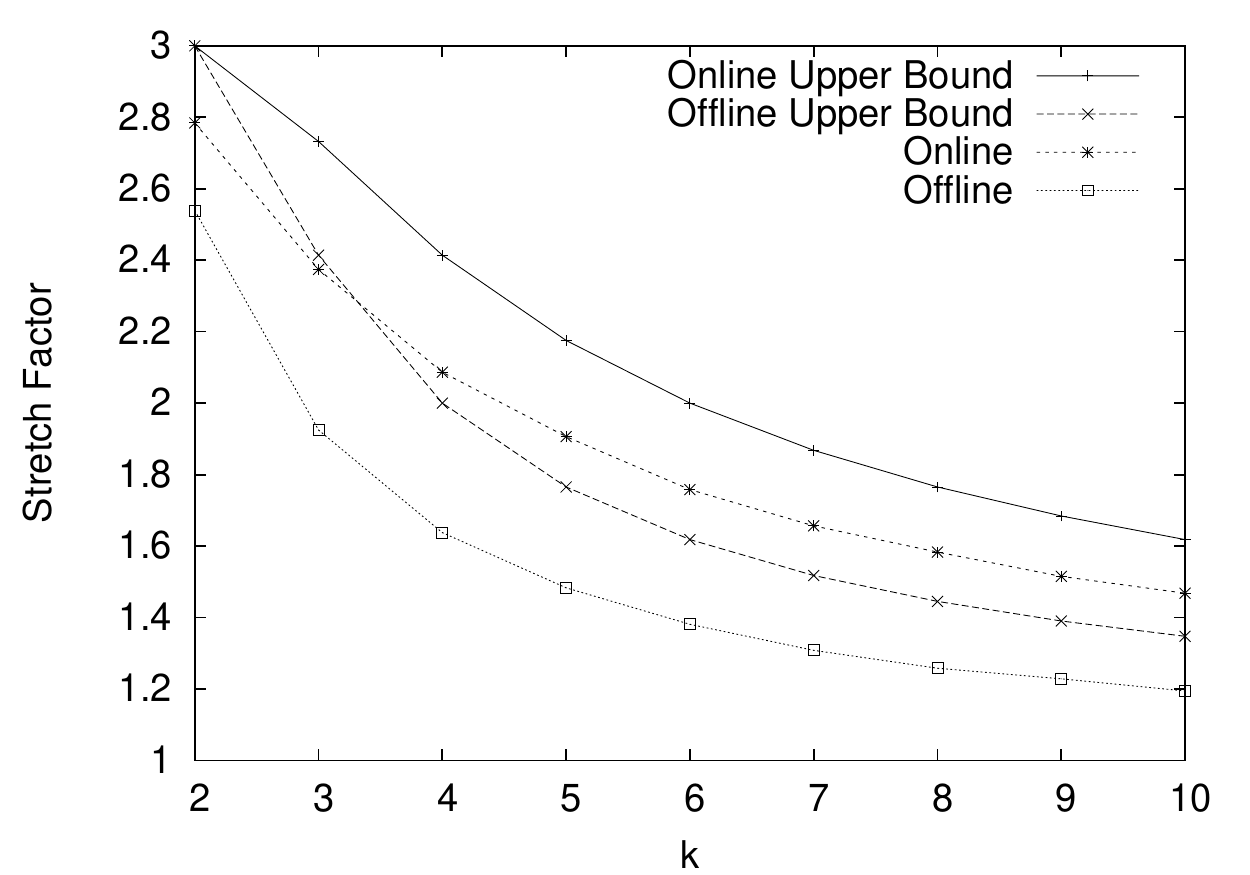} 
\begin{center}
\begin{tabular}{|l|ll|l|ll|}
\hline
$k$ & offline & online & $k$ & offline & online \\
\hline
2 & 2.5390 & 2.7844 & 7 & 1.3079 & 1.6563\\
3 & 1.9245 & 2.3743 & 8 & 1.2579 & 1.5833\\
4 & 1.6377 & 2.0866 & 9 & 1.2283 & 1.5149\\
5 & 1.4831 & 1.9062 & 10 & 1.1945 & 1.4677\\
6 & 1.3809 & 1.7579 & & &\\
\hline
\end{tabular}
\end{center}
%
\caption{Simulation results for 200 nodes using Algorithm~\ref{alg-offk} and Algorithm~\ref{alg-onk}.}\label{fig-simres-chromatic-200}
\end{figure}

\section{Conclusion}\label{section-conclusion}

In this paper, we investigated the problem of computing a spanner 
of a point set that has chromatic 
number $k$. To the best of our knowledge, this problem has not been 
considered before. For small values of $k$ ($k\leq 4$), we provided 
tight upper and lower bounds on the smallest possible stretch factor 
of such spanners. For larger values of $k$, we provided general 
upper and lower bounds which, unfortunately, are not tight. Our construction algorithms
show how to color a point set with $k$ colors such that the complete $k$-partite graph
induced by this coloring has the stated stretch factor. The number of edges in these
graphs can be reduced from quadratic to linear with a slight increase in the spanning ratio by applying
the general technique of Gudmundsson  \emph{et al.}~\cite{glns-adogg-02}. An interesting open problem in this setting of the problem is to find tight upper and lower bounds when $k>4$.

We also considered an on-line variant of this problem where the points
are presented sequentially and our algorithm 
assigns a color to each point upon reception such that the complete $k$-partite graph
induced by the coloring is a constant spanner. Again, for small values of $k$ ($k\leq 4$), we provided 
tight upper and lower bounds on the smallest possible stretch factor 
of such spanners and for $k>4$, we provided general 
upper and lower bounds that are not tight.
A linear-sized spanner can be constructed after all the points have been colored by applying the technique of Gudmundsson  \emph{et al.}~\cite{glns-adogg-02}. However, in this case, our algorithm for computing the linear-sized constant spanner is {\em not} on-line.
Therefore, there are two open problems in the on-line setting. First, to close the gap between the upper and lower bound for $k>4$. Second, provide an on-line algorithm
that computes the linear-sized constant spanner.

\bibliographystyle{plain}
\bibliography{bib}

\end{document}